\documentclass[letterpaper, 10 pt, conference]{ieeeconf} 

\IEEEoverridecommandlockouts 

\usepackage{amsmath} 
\usepackage{amssymb}  
\usepackage{color}
\usepackage{enumerate}
\usepackage{graphicx}

\newcommand{\Rn}{\mathbb{R}^n}

\newcommand{\Rq}{\mathbb{R}^q}
\newcommand{\Rmm}{\mathbb{R}^{m\times m}}
\newcommand{\Rm}{\mathbb{R}^{m}}

\newcommand{\C}{\mathbb{C}}
\newcommand{\Hset}{\mathcal{H}}

\newcommand{\dfn}{:=}
\newcommand{\cl}{\textsc{cl}}
\newcommand{\ob}{\textsc{o}}

\newtheorem{thm}{Theorem}

\newtheorem{lem}{Lemma}

\title{\LARGE \bf Fault-tolerant control under controller-driven
  sampling\\ using virtual actuator strategy*}

\author{Esteban N. Osella$^{1}$, Hernan Haimovich$^1$, and Mar\'{\i}a M. Seron$^{2}$
\thanks{$^*$This work was partially supported by grant PICT 2010-0783, FONCYT-ANPCYT, Argentina.}%
\thanks{$^{1}$ CIFASIS-CONICET and Departamento de Control, Esc. de Ing. Electr\'onica, FCEIA, Universidad Nacional de Rosario, Argentina. {\tt\small \{osella.esteban,h.haimovich\}@gmail.com}}%
\thanks{$^{2}$ Centre for Complex Dynamic Systems and Control, The University of Newcastle, Australia.
 {\tt\small maria.seron@newcastle.edu.au}}%
}

\begin{document}

\maketitle \thispagestyle{empty} \pagestyle{empty}

\begin{abstract}
  We present a new output feedback fault tolerant control strategy for
  continuous-time linear systems. The strategy combines a digital
  nominal controller under controller-driven (varying) sampling with
  virtual-actuator (VA)-based controller reconfiguration to compensate
  for actuator faults.  In the proposed scheme, the controller
  controls both the plant and the sampling period, and performs
  controller reconfiguration by engaging in the loop the VA adapted to
  the diagnosed fault. The VA also operates under controller-driven
  sampling.  Two independent objectives are considered: (a)
  closed-loop stability with setpoint tracking and (b) controller
  reconfiguration under faults. Our main contribution is to extend an
  existing VA-based controller reconfiguration strategy to systems
  under controller-driven sampling in such a way that if objective (a)
  is possible under controller-driven sampling (without VA) and
  objective (b) is possible under uniform sampling (without
  controller-driven sampling), then closed-loop stability and setpoint
  tracking will be preserved under both healthy and faulty operation
  for all possible sampling rate evolutions that may be selected by
  the controller.
\end{abstract}

\section{Introduction}
\label{sec:intro}

Active Fault-Tolerant Control (FTC) systems aim to maintain control
performance levels under a number of fault scenarios, by means of a
controller reconfiguration mechanism. An interesting approach to
controller reconfiguration for FTC is the one based on the concept of
\textit{virtual actuators} (VA) (a complete reference on VA and its
applications and details can be found in
\cite{steffen_controlReconfigurations,lunze_tac06,RiL09,richter_Auto11}). The
main advantage of the VA approach is that it allows the engineer to
design the controller for the nominal (\textit{``healthy''}) plant,
without considering the possible faults. 
More specifically, the method uses a single nominal controller, designed for the
“healthy” system, which is always present in the
closed-loop system, and a virtual actuator, which introduces an interface
between the plant and the controller taking different
actions according to the evaluated fault situation of the plant. 
In healthy operation the virtual actuator is inactive and the whole
control action is provided by the nominal controller. In faulty
operation the virtual actuator generates additional signals that
combine with the existing signals in specific ways in order to cancel
or mitigate the effect of the fault in the closed-loop system. The
advantage of this approach is that any existing nominal controller
which has been designed to satisfy the desired specifications for the
fault-free plant, can be kept in the loop at all times.
In addition, the design of the virtual actuator (which has to adapt to
each type of detected fault) is independent of the controller and is
aimed at preserving specific closed-loop properties in the presence of
faults as, for example, stability and setpoint tracking.


Currently, many control systems involve some kind of shared network
environment with limited bandwidth. Such systems are usually referred
to as Networked Control Systems (NCS) (see the special issues
\cite{antbal_tac04,antbai_pieee07}). Since the network may be shared
among processes, then sampling and acting over the system while
keeping a constant rate may be difficult because it introduces a
trade-off between requiring too much bandwidth and hence restricting
other processes from accessing the network (sampling at a high rate)
or too little bandwidth and hence reducing control performance
(sampling at a low rate). Thus, much research effort has focused on
designing control strategies for systems under varying sampling rate
(VSR).

In the present paper, we consider a type of VSR where a central
controller may be in charge not only of computing feedback but also of
administering access to the shared network. In this setting, the
controller may perform on-line variations of the sampling rate in
order to accommodate for the bandwidth requirements of the different
processes being controlled. This setting is akin to that described in
\cite{cereke_rts02} and has been addressed in previous work by some of
the authors \cite{haimoose_nahs12, haimo_ifac2011,osehaimo_rpic2011,
  osehaimo_aadeca2012}. 
%
%
%
In the current work, we combine this specific type of VSR with an extension of the VA-based controller reconfiguration strategy of \cite{seron_bankVA_ifac11} to the VSR setting, yielding a control scheme as illustrated in Figure~\ref{fig:scheme}. In the proposed scheme, the VSR controller is designed for the nominal or fault-free plant, and hence knowledge of the fault scenario is not needed at the control design stage. The fault tolerance mechanism is responsible for achieving correct closed-loop control under faults. This mechanism requires knowledge of the sampling period command ($h_k$ in Figure~\ref{fig:scheme}) issued by the controller but does not modify the sampling and hold operations. The fault tolerance mechanism considered involves a bank of VAs and a fault detection and isolation (FDI) strategy. Correct closed-loop control under faults is achieved by engaging in the loop the VA from the bank of VAs that is adapted to
the fault diagnosed by the FDI unit. In this paper, we focus on the design and properties of the scheme of Figure~\ref{fig:scheme} to achieve correct setpoint tracking and stability under both nominal and faulty operation, assuming that some FDI strategy is able to correctly diagnose faults. We thus concentrate on controller reconfiguration under faults and leave the integration of an FDI unit as a topic for future work. 
We extend the existing VA-based controller reconfiguration strategy of \cite{seron_bankVA_ifac11} to the VSR case considered, ensuring closed-loop stability and setpoint tracking under nominal and faulty operation regardless of how the VSR controller performs on-line variations of the sampling rate. In this context, our main contribution is to show that the difficulties in this combined scheme are not greater as those for VSR control or uniform-sampling VA-based reconfiguration taken independently.

The proposed scheme can be particularly interesting in circumstances
in which a system with input redundancy is such that its performance
objectives can still be reached under the total loss of some actuators
using a VA technique, and where the tasks related to control and
sampling period selection can be driven by a central controller. 
\begin{figure}[htb]
  \begin{center}
    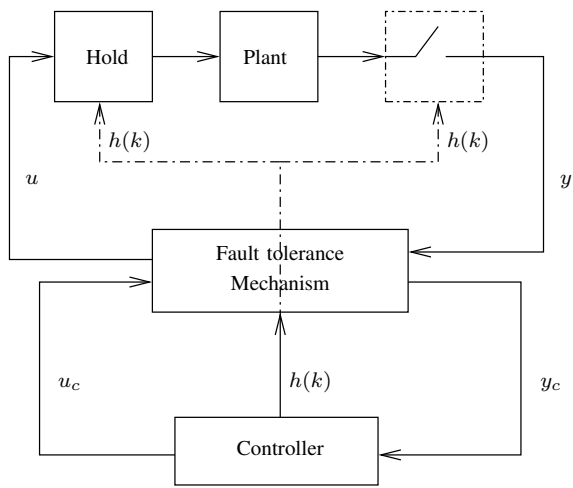
    \caption{Considered scheme: a central controller controls both the
    process and the next sampling instant.}
    \label{fig:scheme}
  \end{center}
\end{figure}

The remainder of the paper is organised as follows. In Section \ref{sec:probform}, we more precisely explain and define the problem. In Section \ref{sec:VA}, we show how the controller and virtual actuators must be designed. In Section~\ref{sec:cl-props}, we provide results that establish correct closed-loop operation. A simulation example illustrating successful operation is provided in Section~\ref{sec:example} and conclusions in Section~\ref{sec:conclusions}.


\section{Problem formulation}
\label{sec:probform}
We consider actuator-fault-tolerant output-feedback control of a
continuous-time plant by means of a discrete-time controller, bank of
virtual actuators, and fault detection and isolation (FDI) unit. The
discrete-time controller is designed for the fault-free (healthy)
situation and hence knowledge of the fault scenario is not necessary
at the controller design stage. In addition to computing the control
action assuming fault-free operation, this fault-ignorant controller
is in charge of performing on-line variations of the sampling rate. In
the sequel, we explain the different components of the feedback
control system considered.

\subsection{Continuous-time plant}
\label{sec:ct-plant}

The plant model that we employ is the following:
\begin{align}
\label{eq:ctstates}
  \dot{x} &= Ax+BF {u},\\
\label{eq:ctoutput}
  {y} &= Cx, \\ 
\label{eq:ctperformance}
  v &= C_vx,
\end{align}
where $x\in\Rn$, $u\in\Rm$ is the control input, $y\in\Rn$ is the
plant measured output, and $v\in\Rq$ is a performance output. As in
\cite{seron_bankVA_ifac11} we represent the operability situation of
the $N=m$ actuators by the matrix $F\in\Rmm$, taking values from a
finite set
\begin{equation}
  \label{eq:faults}
  F\in\mathcal{F}\dfn\{F_0,F_1,\cdots,F_N\},\quad F_0 = I.
\end{equation}
Under healthy operation, the matrix $F$ in (\ref{eq:ctstates}) takes
the value $F=F_0=I$, so that $B$ in \eqref{eq:ctstates} represents the
``healthy'' plant input matrix. The matrix $F_j$, $j=1,\ldots,N$,
models the total loss of the $j$-th actuator. Hence, $F_j$ is obtained
by setting to $0$ the $j$-th diagonal entry of $F_0=I$.  
We assume that the pairs $(A,BF_i)$ are stabilisable for $i =
0,1,\cdots,N$, $(C,A)$ is detectable, and $A$ is invertible.

The performance output $v$ in (\ref{eq:ctperformance}) and the
different fault situations in (\ref{eq:faults}) must be such that for
every desired constant value $v_{ref}$ of the performance output and
every fault $i$, there exists a constant input value $\bar u_i$ so
that the equilibrium state $\bar x_i$ that corresponds to the constant
input $\bar u_i$ under fault $i$ is such that $C_v \bar x_i =
v_{ref}$, i.e. there exist $ \bar x_i$, $\bar u_i$ such that
\begin{equation}
  \label{eq:sstrack}
  \begin{bmatrix}
    A & B F_i \\C_v & 0
  \end{bmatrix}
  \begin{bmatrix}
     \bar x_i  \\\bar u_i
  \end{bmatrix} =
  \begin{bmatrix}
    0\\ v_{ref}
  \end{bmatrix}
\end{equation}
for all $i\in\{0,1,\ldots,N\}$. Condition~\eqref{eq:sstrack} means
that the plant has sufficient levels of redundancy to admit setpoint
tracking in each fault scenario. 
Note that if for a given $v_{ref}$ and specific fault $i$, no $\bar
u_i$ exists satisfying (\ref{eq:sstrack}), then the performance output
will not converge to $v_{ref}$ under fault $i$, no matter how
sophisticated the fault tolerance mechanism may be. On the other hand,
the VA fault tolerance mechanism has to be designed so that the
required equilibrium values satisfying~(\ref{eq:sstrack}) are achieved
under all possible fault situations; this will be addressed in detail
in Section~\ref{sec:va-design}.



\subsection{Fault-ignorant varying-sampling-rate controller}
\label{sec:fault-ign-contr}

As previously explained, knowledge of the fault scenario is not needed
at the controller design stage. Consequently, controller design is
independent of virtual actuator design. We consider a
healthy-plant-model-based reference-tracking sampled-data controller
given by
\begin{align}
  \label{eq:unominal}
  u_c &= -K^h(\hat{x}-x_{ref}) + u_{ref}, \\
  \label{eq:internalmod}
  \hat{x}^+ &= A^h\hat{x}+B^hu_c+L^h(y_c-C\hat{x}),%
\end{align}
where $u_c$ represents the controller-computed plant input signal,
$y_c$ is the plant output signal supplied to the controller,
$x_{ref},u_{ref}$ are state and input constant reference signals,
respectively, and $\hat{x}$, $\hat{x}^+$ are the current and successor
states of the observer~\eqref{eq:internalmod}.
The matrices $A^h$ and $B^h$ are the discrete-time
equivalents of $A$ and $B$ in (\ref{eq:ctstates}), corresponding to a sampling period $h$,
\begin{align}
  \label{eq:aibi}
  A^h \dfn e^{Ah},\quad&\: B^h\dfn \int_0^h{e^{At}Bdt},
  \\ \intertext{and the reference signals satisfy}
  \label{eq:reference}
  A x_{ref} + B u_{ref} &=0, \qquad C_v x_{ref} = v_{ref}.
\end{align}
The controller may perform on-line variations of the sampling period
$h$, under the constraint that all possible sampling periods are taken
from a finite set:
\begin{equation}
 \label{eq:hset}
 h\in \Hset \dfn \{h_1,\cdots, h_{n_s}\},
\end{equation}
where every $h\in \Hset$ should be non-pathological (see \cite{ChF95}
for further details on pathological sampling). The feedback and
observer gains $K^h$ and $L^h$ employed by the controller may depend
on the sampling period selected. The computation of these gains will
be explained in Section~\ref{sec:contr-des}. If no fault tolerance
mechanism were present, the plant input $u$ would equal the
controller-computed plant input $u_c$ and the plant output supplied to
the controller, $y_c$, would equal the true plant output, $y$, at all
sampling instants. In the presence of the fault tolerance mechanism
discussed in the current paper, the equalities $y=y_c$ and $u=u_c$
will be true only under nominal (healthy) conditions and provided the
fault tolerance mechanism accurately detects that the plant is under
healthy operation. 

\subsection{Nominal plant-controller feedback loop}
\label{sec:nom-plant-contr}

As explained previously, in the considered scenario the controller
does not only apply control, but also determines the sampling instants
(based, for example, on the states of the processes and/or
restrictions over the network). Under nominal conditions, at instant
$t_k$ the controller receives the sample $y_c$ and processes it in
order to compute the required feedback action. To do so, it also
determines the instant $t_{k+1} = t_k+h_k$, with $h_k \in
\mathcal{H}$, at which it will take the next sample and control
action. The plant dynamics at the sampling instants can be written as
\begin{align}
  \label{eq:dtsys}
  x^+ = A^hx+B^hF u,
\end{align}
where $A^h$ and $B^h$ are defined in~\eqref{eq:aibi}, $x=x(t_k)$,
$u=u(t_k)$, and $x^+=x(t_{k+1})$ (the \textit{successor} state).

We observe that if constant reference signals $x_{ref}$ and $u_{ref}$
satisfy condition \eqref{eq:reference} for the continuous-time plant, then the same $x_{ref}$ and $u_{ref}$ satisfy
\begin{equation}
  \label{eq:dtnomss}
  x_{ref} = A^h x_{ref} + B^h u_{ref}
\end{equation}
for all $h\in\Hset$. 
To see this, 
premultiply 
$B^h$ in \eqref{eq:aibi} by $A$, producing
\begin{equation}
  A B^h=A\int_0^h e^{A t}dt B = (A^h-I)B.
\end{equation}
Since $A$ is invertible and commutes with $A^h$, then for every
$h\in\mathcal{H}$
\begin{equation}
 \label{eq:idahbh}
   (I-A^h)^{-1}B^h = -A^{-1}B.
\end{equation}  
Next, we write
\begin{equation}
  \label{eq:xrefurefrelation}
    x_{ref} = -A^{-1}Bu_{ref} =  (I-A^h)^{-1}B^hu_{ref}, 
\end{equation}
from which (\ref{eq:dtnomss}) is obtained. 
%

Since the controller may perform on-line variations of the sampling period, the discrete-time system (\ref{eq:dtsys}), obtained by looking at the continuous-time plant only at the sampling instants can be regarded as a \emph{Discrete-Time Switched System}
(DTSS). Since the sequence of sampling periods selected by the controller is
arbitrary (that is, we do not require a priori knowledge on rules for
this selection), we are interested in establishing closed-loop properties (such as stability and setpoint tracking) that hold irrespective of such sequence. In switched systems terminology, we are interested in establishing closed-loop properties that hold under \emph{arbitrary switching} (see, e.g. \cite{showir_siamrev07,lin09:_stabil}). 
System design in order to achieve the required closed-loop properties under arbitrary switching will be addressed in Section~\ref{sec:VA}.


\subsection{Bank of virtual actuators}
\label{sec:bank-va}

As in \cite{seron_bankVA_ifac11}, we consider a bank of VAs where each of the VAs in the bank is designed to compensate for a specific actuator fault.
The VA corresponding to the $i$-th fault situation $F_i$ is given by
\begin{align}
 \label{eq:vadyn}
  \theta^+_i &= A^h\theta_i+B^hu_c-B^hF_iu_i, \\ 
  \label{eq:ui}
  u_i &= -M_i^h\theta_i+N_i^hu_c,\\
  \label{eq:ytheta}
  y_i &= y + C\theta_i.
\end{align}
The variable $\theta_i$ represents the $i$-th VA internal state, $u_i$
is the $i$-th VA plant input signal and $y_i$ the $i$-th VA output to
be supplied to the controller. How $u_i$ and $y_i$ relate to the true
plant input $u$ and the true output supplied to the controller $y_c$
is explained in Section~\ref{sec:sw-logic}. The 0-th VA corresponds to
nominal operating conditions (healthy or fault-free) and has
\begin{equation}
  \label{eq:mn0}
  M^h_0 = 0\quad\text{and}\quad N^h_0= I,
\end{equation}
for all sampling periods $h$. For $i=1,\ldots,N$, the VAs' internal
matrices $M_i^h$ and $N_i^h$ may depend on the sampling period $h$
selected by the controller. The design of $M_i^h$ and $N_i^h$ is
explained in Section~\ref{sec:va-design}. For future reference, note
that substitution of (\ref{eq:ui}) into (\ref{eq:vadyn}) yields
\begin{align}
  \label{eq:vadynsub}
  \theta_i^+ &= A_i^h \theta_i + B^h (I-F_i N_i^h) u_c,\quad\text{with}\\
  \label{eq:Aih}
  A_i^h &\dfn A^h+B^hF_iM_i^h.
\end{align}
that is, the dynamics of each VA is driven by the controller-computed
plant input signal, $u_c$.

\subsection{FDI and controller reconfiguration mechanism} 
\label{sec:sw-logic}

Controller reconfiguration is achieved through a selector that, in
response to the diagnosed fault situation, interconnects the
appropriate virtual actuator from the bank of virtual actuators with
the controller and the plant.

When an FDI mechanism (not described here) detects that the $j$-th
fault has occurred, the $j$-th virtual actuator is interconnected with
the controller and plant by making $u=u_j$ and $y_c=y_j$. Hence,
whenever the FDI diagnoses that the plant is under healthy operation,
the selector will set $u=u_0$ and $y_c=y_0$. The reconfiguration also
resets the 0-th virtual actuator state $\theta_0$ to zero whenever
healthy operation is detected.

Under the above considerations, we next show that if the plant is
under healthy operation and if the FDI mechanism successfully assesses
the plant's healthy condition, then the plant and controller feedback
loop will operate as if the bank of virtual actuators and
reconfiguration mechanism were not present. From (\ref{eq:ui}) and
(\ref{eq:mn0}), then $u_0 = u_c$. Therefore, if at time $k_0$ the FDI
mechanism detects that healthy operation is restored, then $\theta_0 =
0$ according to the virtual actuator state reset condition, $y_0 = y$
from (\ref{eq:ytheta}), and it follows that $y_c = y_0 = y$ and
$u=u_0=u_c$ at time $k_0$. If the plant continues under healthy
operation and the FDI mechanism continues to successfully assess the
plant's healthy condition, then from (\ref{eq:vadyn}) it follows that
$\theta_0 = 0$ and $y_c = y_0 = y$ and $u=u_0=u_c$ will continue to
hold for time instants $k\ge k_0$ until either a fault occurs or the
FDI mechanism ceases to successfully diagnose the plant's condition.


\section{Controller and virtual actuator design}
\label{sec:VA}

The controller and bank of virtual actuators must be designed so that
the performance output $v$ [see (\ref{eq:ctperformance})] is able to
track a constant reference in closed loop, even if faults occur, and
so that all closed-loop variables remain bounded.

\subsection{Controller design}
\label{sec:contr-des}


Controller design involves the appropriate selection of the matrices
$K^h$ and $L^h$ in (\ref{eq:unominal})--(\ref{eq:internalmod}). In
order for the desired closed-loop properties to hold irrespective of
the sampling period sequence selected by the VSR controller, the
matrices $K^h$ and $L^h$ should be selected so that the closed-loop
matrices
\begin{align}
  \label{eq:ahcl}
   A^{h,\cl}&\dfn A^h-B^hK^h, \\
  \label{eq:aho}
   A^{h,\ob}&\dfn A^h-L^hC,
\end{align}
make the sets $\{ A^{h,\cl} : h\in\Hset\}$ and $\{A^{h,\ob} :
h\in\Hset\}$ stable under arbitrary switching. Stability under
arbitrary switching is equivalent to the existence of a Lyapunov
function common to every matrix in the corresponding sets (see,
e.g. \cite{showir_siamrev07,lin09:_stabil}). In general, this common
Lyapunov function may be not quadratic.


If $K^h$ and $L^h$ exist so that the closed-loop matrices (\ref{eq:ahcl}) on the one hand, and (\ref{eq:aho}) on the other, share a common \emph{quadratic} Lyapunov function (CQLF) for all $h\in\Hset$, then $K^h$ and $L^h$ can be computed via linear matrix inequalities (LMIs) (see, e.g. \cite{daafouz02_stabil,sala05_vsrlmi}). 


In some cases, $K^h$ and $L^h$ can be found so that not only CQLFs
exist, but also additional properties hold for the closed-loop
matrices (\ref{eq:ahcl}) and (\ref{eq:aho}). One such case is when
invertible $T_{\cl}$ and $T_{\ob}$ exist so that $T_\cl^{-1} A^{h,\cl}
T_\cl$ and $T_\ob^{-1} A^{h,\ob} T_\ob$ are upper triangular for all
$h\in\Hset$ (solvable Lie algebra case). Several works address the
computation of $K^h$ (and $L^h$) so that this simultaneous
triangularization is achieved for $A^{h,\cl}$ (and $A^{h,O}$) for
general switched linear systems
\cite{haibra_cdc09,haibra_tac10,haimobras_2010} and specifically for
cases as the current one, where $A^h$ and $B^h$ arise from sampling a
single continuous-time systems at different rates
\cite{haimo_ifac2011,osehaimo_rpic2011,haimoose_nahs12,osehaimo_aadeca2012}. Several
facts make this apparently more restrictive design criterion appealing
in the current context because the chances of successful computation
of $K^h$ and $L^h$ increase when either:
\begin{itemize}
\item the DTSS arises from sampling a single continuous-time system at
  different rates
  \cite{haimo_ifac2011,osehaimo_rpic2011,haimoose_nahs12,osehaimo_aadeca2012},
\item the system has many inputs \cite{haibra_aucc11,haibra_rpic11,HB_TAC2013}.
\end{itemize}
Note that both these situations occur in the current case, the second
one because input redundancy is required for successful trajectory
tracking in the presence of total actuator loss, as explained in
Section~\ref{sec:ct-plant}.

\subsection{Virtual actuator features and design}
\label{sec:va-design}

The bank of VAs, as defined in (\ref{eq:vadyn})--(\ref{eq:ytheta}),
jointly with the controller reconfiguration mechanism endow the
feedback loop with specific features when the plant's fault situation
has been correctly diagnosed. One of these features is known as
``fault hiding'' because the controller variables $u_c$ and $y_c$ are
related in such a way as if a plant under nominal conditions were
connected to the controller.  In order to see this feature, define
\begin{align}
  \label{eq:xidef}
  \xi_i \dfn x+ \theta_i, \quad i = 1,\cdots,N,
\end{align}
and write, using \eqref{eq:vadyn} and \eqref{eq:dtsys},
\begin{align}
  \label{eq:xidyn}
  \xi _i^+ = A^h\xi _i+B^h(Fu-F_iu_i)+ B^hu_c.
\end{align}
When the plant fault situation is correctly diagnosed, we have
$F_i=F$, $u_i=u$ and $y_c = y_i$. From the latter equalities,
(\ref{eq:ctoutput}), (\ref{eq:ytheta}) and (\ref{eq:xidyn}), it
follows that
\begin{align}
  \label{eq:ximatched}
  \xi _i^+ &= A^h\xi _i+ B^hu_c,\\
  \label{eq:yc}
  y_c &= C\xi_i. 
\end{align}
Eqs.~(\ref{eq:ximatched})--(\ref{eq:yc}) show that the controller
effectively sees a nominal plant, whose state is $\xi_i$ instead of
$x$.

A second feature of the bank of VAs and switching mechanism is that
the desired setpoint $v_{ref}$ for the performance output $v$ defined
in~\eqref{eq:ctperformance} should be preserved for all fault
situations and sampling period variations, provided the plant fault
situation is correctly diagnosed. In closed loop, the boundedness of
all variables and the tracking of the desired setpoint $v_{ref}$ are
achieved by ensuring the following:
\begin{enumerate}[a)]
\item the controller-computed plant input $u_c$ converges to the
  steady state value $\bar u_c = u_{ref}$, \label{item:ucuref}
\item the VA state vector $\theta_i$ converges to a constant
  steady-state value $\bar\theta_i$, \label{item:xthetass}
\item Under fault $i$, the plant state $x$ and input $u$ both converge to steady-state values $\bar x_i$ and $\bar u_i$ (independent of $h$) and satisfy  (\ref{eq:sstrack}).\label{item:uibar}
\end{enumerate}
In Section~\ref{sec:cl-props} we will show that
items~\ref{item:ucuref})--\ref{item:uibar}) above will be true if we
select the matrices $M_i^h$ and $N_i^h$ as explained next.

The matrices $M_i^h$ should be selected so that 
for every $i \in
\{0,1,\ldots,N\}$, the matrices in the set $\{A_i^h : h \in \Hset\}$,
with $A_i^h$ as in (\ref{eq:Aih}), are stable under arbitrary
switching. The latter can be achieved using, for example, LMI- or
Lie-algebraic-solvability-based methods, as mentioned in
Section~\ref{sec:contr-des}, and implies that every $A_i^h$ is Schur.


Once the $M_i^h$ are designed, select one sampling period $h'\in\Hset$
and compute
\begin{align}
  \label{eq:firstnih}
  N_i^{h'} &= \big[X_i^{h'}\big]^\dagger C_v(I-A_i^{h'})^{-1}B^{h'}, \\
  \label{eq:Xih}
  X_i^h &:= C_v(I-A_i^{h})^{-1}B^{h}F_i\quad\text{for all }h\in\Hset,
\end{align}
where $^\dagger$ denotes the Moore-Penrose generalised inverse. For
every other sampling period $h\in\Hset$, select the corresponding
$N_i^h$ as follows:
\begin{align}
  \label{eq:Nselect}
  N_i^h &= N_i^{h'} - (M_i^{h'} - M_i^{h}) P_i^{h'},\quad \text{where}\\
  \label{eq:Pi}
  P_i^h &:= (I-A_i^h)^{-1}B^h(I-F_i N_i^h)\quad\text{for all }h\in\Hset.
\end{align}
The following result concerning the expression for $P_i^h$ above will
be required in Section~\ref{sec:cl-props}.
\begin{lem}
  \label{lem:imAihBh}
  Let $h\in\Hset$ and $i\in\{0,\ldots,N\}$. Then,
  \begin{equation}
    \label{eq:IdAihBh}  
    (I-A_i^{h})^{-1}B^{h}= -(A+BF_iM_i^{h})^{-1}B. 
  \end{equation}
\end{lem}
\begin{proof}
  Using \eqref{eq:Aih} and \eqref{eq:idahbh}, it follows that
  \begin{align*}
    (I-A_i^{h})^{-1}B^{h} =(I-A^{h}-B^{h}F_iM_i^{h})^{-1}B^{h}\\
    =\big[(I-A^{h})(I-(I-A^{h})^{-1}B^{h}F_iM_i^{h})\big]^{-1}B^{h} \\
    =[A^{-1}(A+BF_iM_i^{h})]^{-1}(I-A^{h})^{-1}B^{h}\\
    =(A+BF_iM_i^{h})^{-1}A[-A^{-1}B]
  \end{align*}
  whence (\ref{eq:IdAihBh}) follows.
\end{proof}

In the next section, we show that if the $M_i^h$ and $N_i^h$ are
selected as previously explained, then items
\ref{item:ucuref})--\ref{item:uibar}) above will be ensured and the
closed-loop system will successfully track the desired setpoint
$v_{ref}$ under both nominal and faulty conditions, even when the VSR
controller performs on-line variations of the sampling period.

\section{Closed-loop properties under VSR}
\label{sec:cl-props}
In this section, we present the main results of the paper. These results are given below as Theorems~\ref{thm:ucbar}, \ref{thm:VAstateconv} and \ref{thm:1}. Each theorem establishes the validity of one of the items~\ref{item:ucuref})--\ref{item:uibar}) detailed in Section~\ref{sec:va-design}, under the design conditions and assumptions explained in Sections~\ref{sec:probform} and~\ref{sec:VA}. These
results ensure the appropriate operation of the VA for the considered
VSR case, by ensuring the boundedness of all closed-loop variables and the convergence of the performance output to the desired reference value under persistent faults.

\subsection{Control-computed plant input convergence}
\label{sec:convergence}
To proceed with our first main result, let us define the following
observer and tracking errors
\begin{align}
  \label{eq:tildexi}
    \tilde{\xi}_i &\dfn \xi_i - \hat{x}, \\
  \label{eq:zeta}
    \zeta_i &\dfn \xi_i-x_{ref},
\end{align}
with $\xi_i$ as in (\ref{eq:xidef}), and express the
controller-computed plant input $u_c$ in (\ref{eq:unominal}) as
\begin{equation}
 \label{eq:ucthetazeta}
  u_c = -K^h \zeta _i +K^h\tilde{\xi}_i+u_{ref}.
\end{equation}
We next establish item~\ref{item:ucuref}) of
Section ~\ref{sec:va-design}. This is done in the following Theorem.
\begin{thm}
  \label{thm:ucbar}
  Consider the continuous-time plant
  (\ref{eq:ctstates})--(\ref{eq:ctoutput}) with VSR controller
  (\ref{eq:unominal})--(\ref{eq:reference}) and bank of VAs
  (\ref{eq:vadyn})--(\ref{eq:mn0}). Suppose that there exist feedback
  matrices $K^h$ and observer-gain matrices $L^h$ as requested in
  \ref{sec:contr-des}. 
  If the plant's fault
  condition is persistent and successfully diagnosed by the FDI unit, then
  \begin{enumerate}[i)]
  \item the combined plant-VA state $\xi_i$ in (\ref{eq:tildexi}),
    where $i$ identifies the plant's fault condition, converges to the
    steady-state value $\bar\xi_i = x_{ref}$, and so does the observer
    state $\hat x$.\label{item:xitzetaconv}
  \item the controller-computed plant input $u_c$ converges to the
    steady-state value $\bar u_c = u_{ref}$.\label{item:ucconv}
  \end{enumerate}
\end{thm}
\begin{proof}
  The equality (\ref{eq:ucthetazeta}) is valid for all
  $i\in\{0,\cdots,N\}$.  Using
  \eqref{eq:unominal}--\eqref{eq:internalmod},
  \eqref{eq:xrefurefrelation}, \eqref{eq:xidef}, and \eqref{eq:yc}, we
  obtain
  \begin{align}
    \label{eq:ximod}
      \tilde{\xi}^+_i &= A^h \tilde{\xi}_i+B^h(Fu-F_iu_i)-L^h(y_c-C\hat{x})\\
    \label{eq:zetamod}
    \zeta_i^+ &= (A^h-B^hK^h)\zeta_i +
      B^h(Fu-F_iu_i)+ B^hK^h\tilde{\xi}_i.
  \end{align}
  By hypothesis, the FDI unit correctly diagnoses the plant's fault
  condition and hence interconnects the $i$-th VA with the controller and plant ($u=u_i$, $F=F_i$, $y_c = y_i$). The error dynamics 
  \eqref{eq:ximod}--\eqref{eq:zetamod} hence become
  \begin{align}
    \label{eq:ximodsel}
    \tilde{\xi}^+_i &= (A^h-L^hC)\tilde{\xi}_i,\\
    \label{eq:zetamodsel}
    \zeta_i^+ &= (A^h-B^hK^h)\zeta_i+B^hK^h\tilde{\xi}_i.
  \end{align}
  Since both $A^h - L^h C$ and $A^h-B^hK^h$ are stable under arbitrary
  switching, then
  \begin{equation}
    \label{eq:lim0}
    \lim_{k\to\infty} \tilde\xi_i = 0\quad\text{and}\quad
    \lim_{k\to\infty} \zeta_i = 0,
  \end{equation}
  which establishes \ref{item:xitzetaconv}). From
  \eqref{eq:ucthetazeta}, then
  \begin{equation}
    \label{eq:ucbaruref}
    \bar u_c = \lim_{k\to\infty} u_c = u_{ref},
  \end{equation}
  which establishes \ref{item:ucconv}). Note that both \eqref{eq:lim0}
  and \eqref{eq:ucbaruref} are true for every possible evolution of
  the sampling periods $h\in\Hset$ (even when varied on-line).
\end{proof}

\subsection{VA state convergence}
\label{sec:const-ss-vsr}
The convergence of the VA state, as per item~\ref{item:xthetass}) of
Section~\ref{sec:va-design}, is established in
Theorem~\ref{thm:VAstateconv} below. We require the following
auxiliary result.
\begin{lem}
  \label{lem:nihprops}
  Consider the matrices $X_i^h$ as defined in \eqref{eq:Xih}. Suppose
  that the continuous-time system matrices $A$, $B$ are such that
  $(A,BF_i)$ are stabilisable for $i = 0,1,\cdots,N$, and for each
  fault matrix $F_i$ there exist constant values $\bar x_i$ and $\bar
  u_i$ satisfying (\ref{eq:sstrack}). Then, $X_i^h \big[ X_i^h
  \big]^\dagger = I$.
\end{lem}
\begin{proof} The existence of constant values $\bar x_i$ and $\bar
  u_i$ satisfying~\eqref{eq:sstrack} is equivalent to the condition
  that the matrix $ \begin{bmatrix} -A & B F_i \\-C_v & 0
  \end{bmatrix}$ has rank $n+q$.  Under non-pathological sampling $h
  \in \mathcal{H}$ the latter rank condition implies (see, e.g., the
  proof of Lemma~IV.3 in~\cite{jemdav_tac03})\footnote{For clarity, in
    this proof we use a subindex to indicate the dimensions of the
    identity matrices, that is, $I_n$ denotes the $n\times n$ identity
    matrix.}
  \begin{equation}
    \label{eq:1}
    \mathrm{rank} \, \begin{bmatrix} I_n-A^h & B^h F_i \\-C_v & 0
  \end{bmatrix} = n+q.
  \end{equation}
  Correct design of $M_i^h$ (recall Section~\ref{sec:va-design})
  implies that $A_i^h$ defined in~\eqref{eq:Aih} is Schur; then
  $(I_n-A_i^h)$ is invertible and we can write
  \begin{multline}
    \label{eq:2}
    \begin{bmatrix}
      I_n& 0 \\ C_v (I_n-A_i^h)^{-1} & I_q
    \end{bmatrix} \begin{bmatrix} I_n  - A^h & B^h F
_i \\-C_v & 0
  \end{bmatrix}
  \begin{bmatrix}
    I_n & 0 \\-M_i^h & I_q
  \end{bmatrix} = \\
  \begin{bmatrix}
    I_n-A_i^h & B^h F_i \\ 0 & X_i^h
  \end{bmatrix},
  \end{multline}
  where $X_i^h \in \mathbb{R}^{q\times m}$ is defined in
  \eqref{eq:Xih}.  Since the first and third matrices on the left hand
  side (LHS) of \eqref{eq:2} are invertible, it follows (using
  Sylvester's inequality and properties of the matrix rank) that the
  rank of the second matrix on the LHS is equal to the rank of the
  matrix on the right hand side of \eqref{eq:2}.  Using~\eqref{eq:1}
  we then have $\mathrm{rank} \, X_i^h =q$, that is, $X_i^h$ has full
  row rank.  Thus, its Moore-Penrose generalised inverse
  $[X_i^h]^\dagger = [X_i^h]^T\big[X_i^h [X_i^h]^T\big]^{-1}$ exists
  and satisfies $X_i^h \big[ X_i^h \big]^\dagger = I_q$. The result
  then follows.
\end{proof}
\vspace{5mm} 

We are now ready to establish item~\ref{item:xthetass})
of Section \ref{sec:va-design}.
\begin{thm}
  \label{thm:VAstateconv}
    Under the hypotheses of Theorem~\ref{thm:ucbar}, consider the
    performance output \eqref{eq:ctperformance}, and
    suppose that for each fault matrix $F_i$, $i=0,1,\ldots,N$, there
    exist constant values $\bar x_i$ and $\bar u_i$ satisfying
    \eqref{eq:sstrack}, and matrices $M_i^h$ so that $\{A_i^h :
    h\in\Hset\}$, with $A_i^h$ as in \eqref{eq:Aih}, is stable under
    arbitrary switching. If $N_i^h$ are selected as explained in
    Section~\ref{sec:va-design} and if the plant's fault condition is
    persistent and successfully diagnosed by the FDI unit, then
    \begin{equation}
      \label{eq:limtheta}
      \lim_{k\to\infty} \theta_i = \bar\theta_i\quad\text{and}\quad C_v\bar\theta_i = 0.
    \end{equation}
\end{thm}
\begin{proof}
  From Theorem~\ref{thm:ucbar}-\ref{item:ucconv}), we know that $\bar
  u_c = u_{ref}$. Let $\bar\theta_i^h$ denote the equilibrium value of
  the VA state $\theta_i$ if a constant sampling period $h$ were kept
  by the controller. Solving from (\ref{eq:vadynsub}) and using
  (\ref{eq:Pi}), we can write
  \begin{equation}
    \label{eq:thetaih}
    \bar\theta_i^h = P_i^h u_{ref}.
  \end{equation}
  We next show that $P_i^h$ is independent of $h$.  Let $h'\in\Hset$
  be the sampling period selected for the computation of $N_i^{h'}$ as
  in (\ref{eq:firstnih}). Using \eqref{eq:Pi} and
  Lemma~\ref{lem:imAihBh}, we can write
  \begin{equation}
    \label{eq:newPi}
    P_i^{h} = -(A+BF_iM_i^{h})^{-1} [B-BF_i N_i^{h}]
  \end{equation}
  for all $h\in\Hset$. Replacing $N_i^{h}$ by the expression
  \eqref{eq:Nselect}, adding $-AP_i^{h'}+AP_i^{h'}$ inside the square
  brackets, and operating, yields
  \begin{multline}
   \label{eq:Pih23}
     P_i^{h} = -(A+BF_iM_i^{h})^{-1} \\
     \big[B(I-F_i N_i^{h'}) - (A+BF_i M_i^{h}) P_i^{h'} +
     \\(A+BF_iM_i^{h'})P_i^{h'}\big].
  \end{multline}
  Using \eqref{eq:newPi}, then $(A+BF_i M_i^{h'})P_i^{h'} =
  -B(I-F_iN_i^{h'})$. Using the latter expression in (\ref{eq:Pih23})
  yields
   \begin{align*}
     P_i^{h} = -(A+BF_iM_i^{h})^{-1} [- (A+BF_i M_i^{h})P_i^{h'}] = P_i^{h'},
  \end{align*}
  which establishes that $P_i^h$ is independent of $h$. We can thus
  write $P_i^h = P_i$ for all $h\in \Hset$. Therefore, the
  steady-state value $\bar\theta_i^h$ also is independent of $h$, as
  follows from (\ref{eq:thetaih}), and we can write $\bar\theta_i^h =
  \bar\theta_i$ for all $h\in\Hset$. Define the incremental variables
  \begin{align}
    \label{eq:deltatheta}
     \Delta \theta_i &:=  \theta_i - \bar\theta_i \\
     \Delta u_c &:= u_c - \bar u_c = u_c - u_{ref}.
  \end{align}
  Using (\ref{eq:vadynsub}), the VA dynamics in the incremental
  variables can be written as
  \begin{equation*}
    \Delta\theta_i^+ = A_i^h \Delta\theta_i + B^h (I-F_i N_i^h) \Delta u,
  \end{equation*}
  where $\Delta u \to 0$ by Theorem~\ref{thm:ucbar}, and $\{A_i^h : h
  \in \Hset\}$ are stable under arbitrary switching for every
  $i=0,\ldots,N$. It follows that $\Delta\theta_i \to 0$ and hence
  $\lim_{k\to\infty} \theta_i = \bar\theta_i$. Using (\ref{eq:Pi}) and
  (\ref{eq:firstnih})--(\ref{eq:Xih}), we can write
  \begin{equation*}
    C_v P_i = C_v P_i^{h'} = \left(I - X_i^h \big[X_i^h\big]^\dagger \right) C_v (I-A_i^h)^{-1} B^h.
  \end{equation*}
  Using Lemma~\ref{lem:nihprops}, then $C_v P_i = C_v P_i^h = 0$ for
  all $h\in\Hset$. Recalling (\ref{eq:thetaih}), then $C_v
  \bar\theta_i = C_v \bar\theta_i^h = C_v P_i u_{ref} = 0$.
\end{proof}

Theorem~\ref{thm:VAstateconv} shows that the virtual actuator state
converges to a constant steady-state value that is independent of the
sampling periods $h \in \mathcal{H}$ and, in addition, is in the null
space of the performance output matrix $C_v$
[see~\eqref{eq:ctperformance}]. This property is key to achieving the
correct setpoint $v_{ref}$ for the performance output $v$, a property
that is established in the following section.

\subsection{Setpoint tracking}
\label{sec:tracking}
We next present our last result, which establishes
item~\ref{item:uibar}) of Section \ref{sec:va-design}, related to the
setpoint tracking property of the VSR VA introduced.
\begin{thm}
 \label{thm:1}
 Under the same hypotheses as for Theorem~\ref{thm:VAstateconv}, the
 plant state $x$ and the performance output $v$ will satisfy
 \begin{equation*}
   \lim_{k\to\infty} x = \bar x_i \quad\text{and}\quad
   \lim_{k\to\infty} v = v_{ref}.
 \end{equation*}
\end{thm}
\begin{proof}
  From Theorem~\ref{thm:ucbar}-\ref{item:xitzetaconv}), the combined
  plant and VA state $\xi_i$ converges to the steady-state value
  $x_{ref}$ and From Theorem~\ref{thm:VAstateconv}, the VA state
  $\theta_i$ converges to $\bar\theta_i$. Recalling (\ref{eq:xidef}),
  then the plant state $x$ must converge to the steady-state value
  $\bar x_i = x_{ref} - \bar\theta_i$. The performance output thus
  satisfies
  \begin{equation*}
    \lim_{k\to\infty} v = C_v \bar x_i = C_v x_{ref} - C_v \bar
    \theta_i = v_{ref},
  \end{equation*}
  where we have used (\ref{eq:reference}) and (\ref{eq:limtheta}).
\end{proof}
The above result shows that, if the correct fault situation has been
diagnosed and the matching VA has been engaged in the closed-loop
system, then the VA-reconfigured system will achieve the desired
constant setpoint tracking, irrespective of the selected sampling
periods $h \in \mathcal{H}$.


\section{Example}
\label{sec:example}
As an application of the proposed strategy, we revisit the two tanks
example presented in \cite{steffen_controlReconfigurations}. The
considered plant is composed of two interconnected tanks A and B,
where the objective is to control the outflow of tank B, using as
control input the inflow of tank A and the opening of the valve
between them.

The control objective is to keep a constat referenced outflow from
tank B. In the linearised model, controlling the level of tank B is
equivalent to controlling its outflow. Thus, the level of tank B will
be the considered objective. The linearised plant equations are given
by
\begin{align}
  \dot{x} &= Ax+BFu\\
  y &= x, \quad  v = C_vx
\end{align}
where 
\begin{align}
  A\dfn 
  \begin{pmatrix}
    -0.25 & 0\\0.25 & -0.25
  \end{pmatrix},\quad & B \dfn
  \begin{pmatrix}
    1 & -0.5\\0& 0.5
  \end{pmatrix}\\
  C_v \dfn&
  \begin{pmatrix}
    0 & 1
  \end{pmatrix}
\end{align}
In this example, we will consider only the loss of actuator $u_2$
(connecting valve blocked in nominal position), referred as Fault type
2 in \cite{steffen_controlReconfigurations}. Thus,
\begin{equation}
  F\in\mathcal{F} \dfn \left\{
  \begin{pmatrix}
    1 & 0\\ 0 & 1
  \end{pmatrix},\:   \begin{pmatrix}
    1 & 0\\ 0 & 0
  \end{pmatrix} \right\}.
\end{equation}
 The considered sampling periods set is given by
\begin{equation}
  \mathcal{H} \dfn \{ 0.1, \:0.05,\: 0.025\}.
\end{equation}
Note that, while each of the considered sampling periods are multiple
of $h_3$, this is only for simplicity and not required by the proposed
strategy. Using \cite[Algorithm 1]{haimobras_2010}, we are able to
compute feedback matrices such that the closed loop share a common
triangularizing transformation. The algorithm uses a procedure that
computes, if possible, common eigenvectors with stability. To do so,
it requires the definition of two auxiliary values (in this case,
selected as $\epsilon_c = \epsilon_d = 10^{-18}$), the first
associated with the stability limits and the second preventing of
selecting eigenvectors in the image of the input matrix. In this
example, we computed the set of feedback matrices for both, the
controller ($K^h$) and the VA ($M_2^h$). Computation of such matrices
using the corresponding representations of the matrix pairs
$(A^h,B^hF_i)$ yields
\begin{align}
  K^{h_1} = 
  \left(\begin{smallmatrix}  
    9.99  &  9.75\\
    -6.14\times 10^{-2} &  -5.99\times 10^{-2}
  \end{smallmatrix} \right),\\
 K^{h_2} =\left(
\begin{smallmatrix}
  19.99 &  19.75\\
  -6.19\times 10^{-2} &  -6.12\times 10^{-2}
\end{smallmatrix} \right)\\
K^{h_3} =\left(
\begin{smallmatrix}
  39.99  & 39.75\\
  -6.21\times 10^{-2}  & -6.18\times 10^{-2}
\end{smallmatrix} \right)
\end{align}
and matrices $M_2^{h}$
\begin{align}
M_2^{h_1} =-\left(
\begin{smallmatrix}
   11.23&  107.99\\
         0&         0
\end{smallmatrix} \right),& \quad
M_2^{h_2} =-\left(
\begin{smallmatrix}
   21.34 & 233.18\\
         0 &        0
\end{smallmatrix} \right)\\
M_2^{h_3} &=-\left(
\begin{smallmatrix}
   41.39 &485.57\\
    0& 0\end{smallmatrix} \right).
\end{align}
Note that the subindex in $M_i^h$ and $N_i^h$ reffers to the
actuator's fault index considered. Using Eq. \eqref{eq:firstnih} for
$h_1$, we can compute $N_2^{h_1}$ and $P_2$. Then, we are able to
compute matrices $N_2^{h_3}$ and $N_2^{h_3}$ using
\eqref{eq:Nselect}. Their computed values are as follows.
  \begin{align}
    N_2^{h_1} &=\left(
\begin{smallmatrix}
    1.00  & 22.46\\
    0       &  0
\end{smallmatrix} \right)\\
    N_2^{h_2}  &=\left(
\begin{smallmatrix}
    1.00&    2.25\\
    0     &    0\end{smallmatrix} \right)\\
    N_2^{h_3}  &=\left(
\begin{smallmatrix}
    1.00&  -37.86\\
    0   &      0\end{smallmatrix} \right)
  \end{align}
  Since $C=I$, we can choose $L^h = A^h$.

  In order to test the reference tracking properties, we change the
  reference from $x_{2,ref} = 0$ to $x_{2,ref} = 0.05$, and while the
  plant is reaching the new setpoint, we simulate the considered
  actuator fault; after a second setpoint change for $x_2$ from 0.05
  to 0, we simulate the actuator restitution to the healthy
  situation. The resulting responses are shown in Figure
  \ref{fig:example}. Observe that, in both situations the controller
  always tracks the reference provided the VA matches the fault. Also
  observe that, since during the fault the only active actuator is the
  pump, then the only way to keep the level of the second tank at 0.05
  is by increasing $x_1$, the level in the first tank, as shown by the
  blue curve in the top plot of Figure \ref{fig:example}. The fault
  index and the selected sampling period index variation are shown in
  the bottom plot of Figure \ref{fig:example}. We used \cite[Algorithm
    1]{haimobras_2010} setting a closed loop eigenvalue at zero, and
  hence, the response may be aggressive and differ from the one in
  \cite{steffen_controlReconfigurations}.
\begin{figure}[htb]
  \begin{center}
    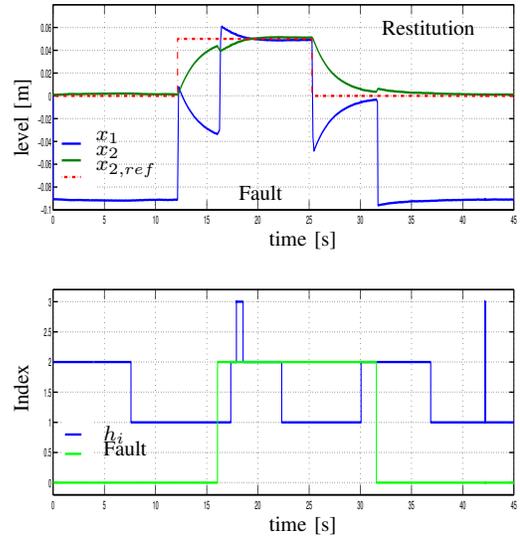
    \caption{Tank levels with a fault in the second actuator occurring
      while the states track a setpoint change, and restitution to the
      healthy situation while the setpoint change is reversed (top
      plot). Sequence of active sampling period index and fault
      indexes (bottom plot).}
    \label{fig:example}
  \end{center}
\end{figure}

\section{Conclusions}
\label{sec:conclusions}
In this paper we have presented a new approach for the virtual
actuator technique under varying sampling rate control systems.  In
this approach, the controller is in charge of both providing the
control action and administering the sampling periods, taken from a
finite set. The considered fault scenario consists of abrupt actuator
outages and we have assumed that correct fault detection and isolation
is provided externally.
The main results of this paper show that in steady state, the control
input will achieve its desired constant reference value, the VA states
will converge to a constant value (irrespective of the sampling period
used, and the variations on it), and the desired constant setpoint
tracking objective is ensured for the performance variable.
The considered approach can be of particular interest for plants with
input redundancy controlled using a NCS. Future work will focus on the
design of an automatic fault detection and isolation method for the
current approach.


\bibliographystyle{plain} 


\end{document}